\documentclass[]{interact}

\usepackage{epstopdf}
\usepackage[caption=false]{subfig}
\usepackage{tikz-qtree}
\usepackage{array} 
\usepackage{tikz}
\usepackage{multicol}
\usetikzlibrary{arrows}
\usepackage[numbers,sort&compress]{natbib}
\bibpunct[, ]{[}{]}{,}{n}{,}{,}
\makeatletter
\def\NAT@def@citea{\def\@citea{\NAT@separator}}
\makeatother

\theoremstyle{plain}
\newtheorem{theorem}{Theorem}[section]
\newtheorem{problem}{Problem}

\theoremstyle{definition}
\newtheorem{definition}[theorem]{Definition}
\newtheorem{example}[theorem]{Example}

\theoremstyle{remark}
\newtheorem{remark}{Remark}
\newtheorem{notation}{Notation}
\newcommand{\biggg}{\rotatebox[origin=c]{90}{$\lor$}}
\newcommand{\bg}{\ \biggg\ }
\begin{document}

\articletype{Research Article}

\title{ A Graph Theory Approach for Regional Controllability of Boolean Cellular Automata. }

\author{
\name{S. Dridi\textsuperscript{a,b}, S. El Yacoubi\textsuperscript{a}\thanks{Corresponding author: Samira El Yacoubi. Email: yacoubi@univ-perp.fr}, F. Bagnoli\textsuperscript{b}, A. Fontaine\textsuperscript{c}}
\affil{\textsuperscript{a}Team Project IMAGES ESPACE-Dev, University of Perpignan, France; \textsuperscript{b} Department of Physics and Astronomy and
CSDC, University of Florence, Italy; also INFN, sez. Firenze \textsuperscript{c}  Universit\'e de Guyane - UMR Espace-Dev, Cayenne, France}
}

\maketitle

\begin{abstract}

Controllability is one of the central concepts of modern control theory that allows a good understanding of a system's behaviour. It consists in constraining a system to reach the desired state from an initial state within a given time interval. When the desired objective affects only a sub-region of the domain, the control is said to be regional. The purpose of this paper is to study a particular case of regional control using cellular automata models since they are spatially extended systems where spatial properties can be easily defined thanks to their intrinsic locality. We investigate the case of boundary controls on the target region using an original approach based on graph theory. Necessary and sufficient conditions are given based on the Hamiltonian Circuit and strongly connected component. The controls are obtained using a preimage approach. 
\end{abstract}


\section{Introduction}
Control theory is a branch of mathematics that deals with the behaviour of dynamical systems studied in terms of inputs and outputs. With the recent developments in computing, communications, and sensing technologies, the scope of control theory is rapidly evolving to encompass  the increasing complexity of real-life phenomena. Controllability and observability are two major concepts of control theory that have been extensively developed during  the last two centuries. The concept of controllability refers to the ability of designing control inputs so as to steer the state of the  system to desired values within an interval
time $[0,T]$ while the observability describes whether the internal state variables of the system can be externally measured. These concepts are being increasingly useful in a wide range of  applications such as: biology, biochemistry, biomedical engineering, ecology, economics etc. 
\cite{cury2013systems,intriligator1980applications}. 
Controllable and observable systems have been characterized so far using the Kalman condition in the linear case. The aim of this paper is to find a general way to give a necessary and sufficient condition for controllability of complex systems via cellular automata models. We  concentrate in this work on regional controllability via boundary actions on the target region $\omega$ that consists in achieving an objective only in a subdomain of the lattice  when some specific actions are exerted on the target region boundaries.

The concept of controllability has been widely studied  for both finite \cite{sontag2013mathematical,liu2009elementary} and infinite dimensional systems \cite{street1995analysis}. 
As in many practical problems one is interested in achieving some objectives only on a 
restricted given sub-region, the notion of controllability has been extended to the so-called regional controllability concept that has been  introduced by El Jai and Zerrik in
\cite{zerrik2000} and well studied in several works
\cite{el1995regional,zerrik2004regional,lions1986exact}.
For distributed parameter systems, the term regional has been used to refer to control problems in which
the desired state is only defined and may be reachable on a portion of the domain $ \Omega $. 

As for controllability issue, one can consider controls applied on the boundary of the domain  or, in the case of regional controllability,   on the
boundaries of the considered subregion. The controls 
will steer the system from an initial state to a desired target on a subregion $\omega$ during a fixed time $T$.

Boundary regional controllability problems for distributed parameter systems have been mainly described so far,  by partial differential equations and considered for linear or nonlinear, continuous or discrete systems. In this paper, we propose to investigate  these  problems by using cellular automata as they have been often  considered as a good alternative to partial differential equations. 
\cite{toffoli1984cellular,green1990cellular}. 

Cellular automata (CA for short) are discrete dynamical systems considered as the simplest models of spatially extended systems. They are  widely used for studying the mathematical properties of discrete systems and for modelling physical systems, \cite{chopard2012cellular,toffoli1984cellular}. However, control of systems described by CA remains very difficult. The techniques for controlling discrete systems are quite different from those used in continuous ones, since discrete systems are in general
strongly nonlinear and the usual linear approximation
cannot be directly applied. We restrict our study to the case of some CA rules.

A Boolean cellular automaton is a collection of cells that can be in one of two states, on and off, 1 or 0, $S=\{0,1\}$. The states of each cell varies in time depending on the states of their neighbourhood and a local transition function that defines the connections between  the cells. This function can be deterministic   or probabilistic,  synchronous or asynchronous,  linear or nonlinear. 

We focus in this study, on a particular case of  deterministic and synchronous  transition rule that calculates the  output of   each cell,   at each time step  as a function of the current state of the cell and the states of its two immediate neighbours.  A CA  configuration or global state defines the image representing the states at time $t$, of the whole lattice cells. The  CA evolution is described by the succession of  configurations at different time steps. 

This evolution in the case of  Boolean deterministic CA can be represented by an oriented graph where the vertices corresponds to the configurations obtained from the binary representation converted to decimal. There is an arc between two vertices $v_{1}$ and $v_{2}$ if the configuration corresponding to $v_{2}$ can be obtained from the configuration obtained from $v_{1}$ where the local transition function is applied.

Some regional control problems has been studied using CA, see for
instance \cite{bagnoli2017toward,el2003analyse,fekih2006regional}. In
\cite{el2002regional}, the control of 1D and 2D additive CA has been
studied by exploring a numerical approach based on genetic
algorithms. The regional control problem has been studied on
deterministic cellular automata in \cite{bagnoli2017toward} and on
probabilistic CA in \cite{bagnoli2018regional}.  In  \cite{sdridi} an approach based  on Markov Chain has been used to prove the  regional controllability of linear 1D and 2D cellular automata instead of using the well known Kalman criterion. 

 In this paper, we pursue our investigation on regional control problems of deterministic CA by using a graph theory approach. The evolution of a controlled CA can be represented by an oriented graph where the vertices represent the possible configurations in the controlled region $\omega$ which are related to each other by arcs. A couple of vertices $(v_1,v_2)$ (CA configurations restricted to $\omega$) are related by arc if it exists a boundary control $(\ell,r)$  such that $v_2$ is reachable starting from $v_1$. In this paper we focus on the problem of regional controllability by applying the controls on the boundaries of the region $\omega$. We prove the regional controllability by checking the existence of Hamiltonian Circuit which  allows us to give a sufficient condition and necessary condition. We address also the problem of decidability of regional controllability by looking for a strongly connected component in the graph related to the controlled CA. A necessary and sufficient condition for regional controllability is then obtained. Finally, the controls required on the boundaries of $\omega$ that ensure  the regional controllability are obtained using a method for generating the preimages.

The paper is organized as follows.  Section 2 provides necessary definitions and  section 3 presents the problem of regional controllability. Section 4 is devoted to the formulation of the problem using transition graphs and section 5 gives necessary and sufficient  conditions for regional controllability. It first deals with the existence of a Hamiltonian Circuit in the graph representing the Boolean CA global evolution and then  the decidability criterion of regional controllability
by establishing a relation with strongly connected component is given.  According to this criterion, we give a classification of selected rules in the one-dimensional CA case.  In section 6, we introduce a method to trace the configurations where a regional control is possible using a method based on preimages. Finally, a conclusion will be given in section 7. 

\section{Basic definitions}

\begin{definition}\cite{el2003analyse}
  A cellular automaton (CA for short) is defined by a tuple
  $ A=( \mathcal{L}, S, \mathcal{N},f) $ where: \\
  \begin{enumerate}
      \item $ \mathcal{L} $ is a cellular space which consists in a
        regular paving of domain $ \Omega $ of $ \mathbb{R}^{d} $, 
        $d=1,2 $.

      \item $ S $ is a finite set of possible states.
      \item $ \mathcal{N} $ is a function that defines the
        neighborhood of a cell $ c $. We denote:
	\begin{eqnarray*}
	  \mathcal{N}: \mathcal{L} & \rightarrow & \mathcal{L}^{r}\\ c
          & \rightarrow & \mathcal{N}(c)=(c_{i_{1}},c_{i_{2}},\dots,c_{i_{r}})
	\end{eqnarray*} 
	where $ c_{i_{j}} $ is a cell for $ j=1,\dots,r $ and $ r $ is the
        size of the neighborhood $ \mathcal{N}(c) $ of the cell $c$.

      \item $ f $ is the transition function that allows to compute
        the state of a cell at time $ t+1 $ according to the state of
        its neighborhood at time $ t $. It is defined as follow:
	\begin{eqnarray*}
	  f: S^{r} & \rightarrow & S\\ s_{t}(\mathcal{N}(c)) &
          \rightarrow & f(s_{t}(\mathcal{N}(c))) = s_{t+1}(c)
	\end{eqnarray*} 
	where $ s_{t}(c) $ is the state of a cell $ c $ at time $ t $
        and $ s_{t}(\mathcal{N}(c))=\{s_{t}(c^{'}),c^{'}\in
        \mathcal{N}(c)\} $ is the state of the neighborhood of $c$.
  \end{enumerate}
\end{definition}

\begin{definition}
  
  \begin{itemize}
    \item The configuration of a CA at time $ t $ corresponds to the
      set $\{s_{t}(c),c \in \mathcal{L}\} $.
 
    \item The global dynamics of a CA is given by the function:
      \begin{eqnarray*}
        F: S^{\mathcal{L}} & \rightarrow & S^{\mathcal{L}}\\ \{s_{t}(c),c
        \in \mathcal{L}\} & \rightarrow & \{s_{t+1}(c),c \in \mathcal{L}\}
      \end{eqnarray*} 
      
      F maps a configuration of CA at time $ t $ a new
      configuration at time $ t+1 $.
    \item We denote by $\omega$ the region we want to control. We
      have: $\omega=\{c_1,\dots,c_{n}\}$.
  \end{itemize}
\end{definition}

\begin{definition}
 An elementary CA (ECA) is a one-dimensional cellular automaton constituted by an array of cells  that take their
 states in $\{0,1\}$ and change it depending only on the states of their
 neighbours. One can say that the neighborhood of a cell is given by the cell itself and its right and left nearest neighbours. \\

Since there are $2^3=8$ possible binary states for a cell and its two immediate neighbours, there are a total of $2^8=256$ ECA, each of which can be indexed with an 8-bit binary number \cite{wolfram1983,Wolfram2002}.

 
\end{definition}
\section{Problem Statement}
Let us consider:
\begin{itemize}
    \item a 1D-cellular domain $\mathcal{L}$  of $N$ cells,  
    \item a discrete time horizon $I=\{0,1, . . . , T\}$,
   \item  a sub-domain $\omega$ that defines the controlled  region where  we want to drive the CA towards a given configuration.

   
   It will  contain  $n$ cells denoted by $c_{i}$,  $i = 1, 2, \cdots, n$, $n < N$.
   \item $\overline{\omega}=\{c_{0},c_{1},\dots,c_{n},c_{n+1}\} = \omega \cup \{c_{0},c_{n+1}\}$  where  $\{c_{0},c_{n+1}\}$ are the boundary cells of $\omega$ where we apply control.
\end{itemize}
 
We are interested in the problem of regional controllability via actions exerted  on the boundary of  the target region $\omega$ which is a part of the cellular automaton  space as illustrated in Figure \ref{fig1} for $n=3$. Our aim is to determine the values  at the boundary cells in order to obtain a  specific behaviour on $\omega$.

\begin{figure}[!ht]
\centering
{%
\resizebox*{9cm}{!}{\includegraphics{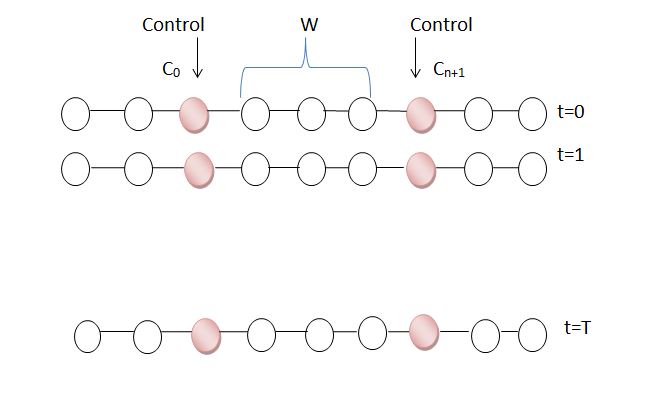}}}
\caption{Regional control of one-dimensional CA.} 
\label{fig1}
\end{figure}

\begin{definition}\cite{el2002regional}
 Let $s_{d} \in S^{\omega}$ ($s_d: {\omega} \rightarrow S$) be a desired profile to be reached on $\omega  \subset \mathcal{L}$. The CA is said to be regionally controllable for $\omega$ at time $T$ 
 if there exists a control sequence  $ u=(u_{0},\dots,u_{T-1}) $ where
 $u_i=(u_i(c_{0}),u_i(c_{n+1})), \quad i=0,\dots,T-1$ such as:
  \begin{displaymath}
    s_{T}=s_{d} \quad \text{on} \quad \omega
  \end{displaymath}
  where $ s_{T} $ is the final configuration at time $ T $ and $s_{d}$   is the desired configuration.
\end{definition}


\begin{notation}
  Let us introduce the following notations:\\
  \begin{eqnarray*}
    \ell^t=u_t(c_{0}) \\
    r^t=u_t(c_{n+1})\\
 s^t_i=s_t(c_{i}) 
 \end{eqnarray*}
$\forall~i.~ 1 \le i \le n$. \\

  $(l \cdot x \cdot r)$ is the concatenation operation
    describing the CA state on $\overline{\omega}$
    where $x=s(c_{1}),\dots,s(c_{n})$, $\ell = u(c_{0})$ and $r = u(c_{n+1})$.
\end{notation}

\begin{problem}
  Starting from an initial condition and for a given  desired configuration $s_d$, the considered problem of regional controllability   consists in finding the control required on the boundaries   $\{c_{0},c_{n+1}\}$, in order to get at time $T$, the  configuration $s_d$ in the   controlled region $\{c_{1},\dots,c_{n}\}$, such that
  $s_{d}(c_{i})=s_{T}(c_{i}) $ $ \forall i=1,\dots,n $, for a given 
  time horizon $T$. 
\end{problem}

\begin{example}\label{ex:1}
 Consider the  Wolfram's rule $90$ for which the  evolution  can completely be described by a table mapping the next state  from all possible combinations of three inputs $(s_{-1},s_0,s_{+1})$ according to the sum modulo 2 of the state values of the cells to its left and to its right $s_{-1} \oplus s_{+1} $: 

  \begin{center}
    \begin{tabular}{cccc}
      $f_{90}$: & $111$ & $\mapsto$ & $0$\\
      & $110$ & $\mapsto$ & $1$\\
      & $101$ & $\mapsto$ & $0$\\
      & $100$ & $\mapsto$ & $1$\\
      & $011$ & $\mapsto$ & $1$\\
      & $010$ & $\mapsto$ & $0$\\
      & $001$ & $\mapsto$ & $1$\\
      & $000$ & $\mapsto$ & $0$\\
    \end{tabular}
  \end{center}

 For instance, with $n=6$ (cf. Figure \ref{fig:ex1D}), if we assume starting  at time $0$ with   an initial configuration 
  $\{s_1^0, s_2^0, s_3^0, s_4^0,s_5^0, s_6^0\} = \{011100\}$ on $\omega = \{c_{1}, \cdots, c_{6}\}$ and given a desired null state on $\omega$, there exists a control $u=(u_0, u_1, u_2)$ where $u_0 = (\ell^0, r^0) = (0,1), u_1 = (\ell^1, r^1) = (1,0), u_2 = (\ell^2, r^2) = (1,0) $ that are applied on cells $c_{0}, c_{7}$,  such that the final CA  configuration on $\omega$ obtained at time $T=3$ from the evolution of rule 90,   is  $\{s_1^3, s_2^3, s_3^3, s_4^3,s_5^3, s_6^3\} = \{000000\}$.
 
  \begin{figure}
  \centering
  \begin{multicols}{2}
\begin{tabular}{l|c|r}
  $0000$ & \color{blue}$011100$ & $0000$ \\
  $0000$ & \color{blue}$110110$ & $0000$ \\
  $0000$ & \color{blue}$110111$ & $0000$ \\
  \hline
  & $110101$ & 
\end{tabular} 
  \begin{tabular}{l|c|r}
  $000\color{red}0$ & $\color{blue}011100$ & $\color{red}1\color{black}000$ \\
  $000\color{red}1$ & \color{blue}$110111$ & $\color{red}0\color{black}100$ \\
  $000\color{red}1$ &\color{blue} $010101$ & $\color{red}0\color{black}010$ \\
  \hline
  & $000000$ &
\end{tabular}
  \end{multicols}
  \caption{The  evolution   of CA Wolfram rule 90 on the region $\omega = \{c_{1}, \cdots, c_{6}\}$ starting with the same initial configuration;  on the left  without control and  on the right with control.}
     \label{fig:ex1D}
\end{figure}

\end{example}

\begin{remark}

The same problem can be defined on  two-dimensional CA.  For example,  we can apply the control on
one side of the boundary or on the whole  boundary cells of the controlled region $\omega$ in order to get the desired state inside $ \omega $ (see Figure
\ref{fig2}).


\end{remark}

\begin{figure}[!ht]
\centering{%
\resizebox*{9.5cm}{!}{\includegraphics{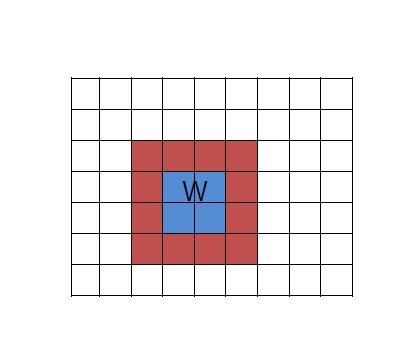}}}
    \caption{Regional Control of two-dimensional CA} \label{fig2}
\end{figure}

  \begin{example}

 Consider the following local evolution rule of a  two dimensional CA  given by the function: 
  \begin{displaymath}
    s^{t+1}(c_{i,j})=s^{t}(c_{i-1,j}) \oplus s^{t}(c_{i+1,j}) \oplus s^{t}(c_{i,j-1}) \oplus s^{t}(c_{i,j+1})
     \label{Rule2D}
  \end{displaymath}

We consider a  controlled region given by the square $\omega=\{c_{1,1},c_{1,2},c_{2,1},c_{2,2}\}$.\\

For a given initial configuration  given by $\{s_{1,1}^0, s_{1,2}^0, s_{2,1}^0, s_{2,2}^0\} = \{1,0,0,0\}$ on $\omega$, we first let the system evolve without applying controls and  get the final configuration $\{s_{1,1}^1, s_{1,2}^1, s_{2,1}^1, s_{2,2}^1\} = \{0,1,1,0\}$ at time $T=1$. \\

We look for  controls applied on the boundary cells of $\omega$ in order to obtain a desired configuration  consisting of all 1s on $\omega$, this can be obtained by controls illustrated in red in Figure~\ref{fig:ex2D}.
  
  \begin{figure}
      \centering
 \begin{minipage}[t]{0.4\textwidth}
    $$\begin{pmatrix} 
       \color{red}0 &  \color{red}0 &  \color{red}0 &  \color{red}0 \\
       \color{red}0 & \color{blue}1 & \color{blue}0 &  \color{red}0\\
      \color{red} 0 & \color{blue}0 & \color{blue}0 &  \color{red}0\\
       \color{red}0 &  \color{red}0 &  \color{red}0 &  \color{red}0
    \end{pmatrix}
    $$
    \centering
    T=0 
    $$
    \begin{pmatrix} 
     - & - & - & - \\
      - & \color{blue}0  & \color{blue}1 & -\\
      - & \color{blue}1 & \color{blue}0 & -\\
      - & - & - & -
    \end{pmatrix}
    $$
    \centering
    T=1
  \end{minipage}
  \begin{minipage}[t]{0.4\textwidth}
   $$\begin{pmatrix} 
       \color{red}0 &  \color{red}1 &  \color{red}1 &  \color{red}0 \\
       \color{red}0 & \color{blue}1 & \color{blue}0 &  \color{red}1\\
      \color{red} 1 & \color{blue}0 & \color{blue}0 &  \color{red}1\\
       \color{red}0 &  \color{red}1 &  \color{red}0 &  \color{red}0
    \end{pmatrix}
    $$
    \centering
    T=0 
    $$
    \begin{pmatrix} 
     - & - & - & - \\
      - & \color{blue}1  & \color{blue}1 & -\\
      - & \color{blue}1 & \color{blue}1 & -\\
      - & - & - & -
    \end{pmatrix}
    $$
    \centering
    T=1
  \end{minipage}
      \caption{Evolution of the CA rule example \ref{Rule2D} on $\omega$ in the autonomous and controlled cases on the left and right matrices respectively.  }
      \label{fig:ex2D}
  \end{figure}

  \begin{remark}
  
  We can define asymmetric controls i.e, we keep a part of the boundaries fixed ( for instance at 0) and act on a subset of the boundary of the controlled region (red cells) in order to  get the desired state (see Figure \ref{asymetric}).

			\begin{figure}[!ht]
\centering
{%
\resizebox*{7cm}{!}{\includegraphics{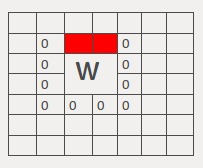}}}
\caption{Regional control of two dimensional CA with asymmetric controls} 
\label{asymetric}
\end{figure}
  \end{remark}
\end{example}		
The above examples of CA in one and two dimensional cases  show that it is possible to steer a system from an initial state to a desired target on a subregion of the domain by acting on its boundary. The aim of this paper is to generalize this results and find necessary and sufficient conditions for  the regional controllability of some Boolean CA rules. The proposed method in the following will be based on transition graphs.

\section{Transition graph approach and regional controllability problem}
\label{sect:transG}

In this section, we describe the main tool on which this paper is based: the transition graph $\Upsilon$. It  was inspired by the one introduced in ref \cite{sdridi} where the authors have built a transformation matrix based on all possible state combinations of the CA to show the transition  steps. Let us describe here the construction of $\Upsilon$ and the transformation matrix $\mathcal{C}$ that is the associate adjacency matrix of $\Upsilon$.\\

   Recall that the evolution of controlled CA for one step can be represented by a directed graph  where the vertices represent the configurations and the arcs  represent the transition from a configuration to another one in one step \emph{i.e.} by applying the global transition function $F$. Consider an Elementary CA where the controlled region $\omega$ is of size $|\omega|$ and controls are applied on its two boundary cells  $ \{c_{0},c_{n+1}\} $. When considering the restriction of $F$ on $\mathcal{S}^{|\omega|}$, there exists a bijection between  $\mathcal{S}^{|\omega|}$ and the set of integers $[0:2^{|\omega|}-1]$ that represents CA configurations on ${\omega}$ as $|\omega|$-bit binary numbers. Let $\lambda$ be a vertex labelling such that for every vertex $v$, $\lambda(v)$ is the Boolean conversion of vertex $v$. 
   
   We define the transition graph $\Upsilon=(V,A)$ as follow where the vertices $V$ corresponds to each possible configuration of the region $\omega$ and $A$ is the set of arcs. Let $v_1$ and $v_2$ be two vertices in $V$,
   there is an arc from the vertex $v_1$ to the vertex $v_2$ if there exists a control $u=(\ell,r)\in
\{(0,0) ; (0,1) ; (1,0); (1,1)\}$ 
such that $\lambda(v_2)$ is equal to 
$ F_{|_{\mathcal{S}^{\overline\omega}}}(\ell\cdot
{\lambda(v_{1})}\cdot r) $,
where the $\lambda(v_{1})$ denotes the configuration in the controlled region at time $t$ and $\lambda(v_{2})$ denotes the configuration in the controlled region at time $t+1$.

We denote by $\mathcal{C}$ the transition matrix which is the associate adjacency matrix of the graph $\Upsilon$. 
The transition matrix is built as a Boolean matrix of size $ 2^{|\omega|} \times 2^{|\omega|}$. 
There is a $1$ at position $(i,j)$, the $i{\mbox{th}}$ row and $j{\mbox{th}}$ column,  if there is an arc between vertices $i$ and $j$ for all
$i,j$ in $[0:2^{|\omega|}-1]$. Otherwise it stays at $0$.\\

We proceed as follow to construct the transition graph $\Upsilon$ (the algorithm is given in the appendix). For each 
vertex $v$, we compute the four configurations (represented by 
$u_1, u_2, u_3, u_4$) obtained by the application of the global transition 
function $F_{|_{\mathcal{S}^{\overline\omega}}}$ to the four possible configurations 
obtained by the concatenation of the controls 
($(0, 0) ; (0,1) ; (1, 0) ; (1, 1)$) on the extremities of $v$. Then we add 
an arc from $v$ to each of the four $u_i$. In total, the time complexity to 
build $\Upsilon$ is $O(|V|)$ \emph{i.e.} $O(2^{|\omega|})$ where $|\omega|$ is the size of the
controlled region $\omega$ in the CA. The space complexity is the size of the Boolean matrix $\mathcal{C}$: $O(|V|\times |V|)=O(2^{|\omega|}\times 2^{|\omega|})$. Note that the number of 
arcs is at most $4\times |V|$. 

\begin{remark}
Note that we have taken the binary representation for the controlled region in a reverse order (the least significant bit is the first one). For instance, for a controlled region of size 3, we note: 
$\lambda(100)=1$, $\lambda(110)=3$ and $\lambda(001)=4$
\end{remark}

\begin{example}\label{exGC2}
  For instance, consider the rule $30$    where the controlled region is of size $|\omega|=2$ for more simplicity. The corresponding graph is represented in Figure~\ref{figGC1}. 

  The corresponding table for rule 30 is:

  \begin{center}
    \begin{tabular}{cccc}
      $f$: & $111$ & $\mapsto$ & $0$\\
      & $110$ & $\mapsto$ & $0$\\
      & $101$ & $\mapsto$ & $0$\\
      & $100$ & $\mapsto$ & $1$\\
      & $011$ & $\mapsto$ & $1$\\
      & $010$ & $\mapsto$ & $1$\\
      & $001$ & $\mapsto$ & $1$\\
      & $000$ & $\mapsto$ & $0$\\
    \end{tabular}
  \end{center}

  Consider vertex $2$ whose binary conversion is $01$. We have that \\

  $F_{|_{\mathcal{S}^{\overline\omega}}}(0010)=F_{|_{\mathcal{S}^{\overline\omega}}}(0011)=11$ and $F_{|_{\mathcal{S}^{\overline\omega}}}(1010)=F_{|_{\mathcal{S}^{\overline\omega}}}(1011)=01$. Therefore
  there are two arcs from $2$: $(2,2)$ and $(2,3)$ as the binary
  conversion of $3$ is $11$.
\end{example}

\begin{figure}[!ht]
    \begin{center}
      \begin{tikzpicture}
        \node[draw,circle] at (0,0) (A){$0$};
        \node[draw,circle] at (0,-2) (B){$1$};
        \node[draw,circle] at (2,-1) (C){$3$};
        \node[draw,circle] at (4,-1) (D){$2$};
        \path[->,>=stealth',shorten >=1pt,auto,
          node distance=2.8cm, semithick]
        (A) edge [bend left] (B)
        (A) edge [bend right=20] (C)
        (A) edge [bend left] (D)
        (A) edge [loop above] (A)
      (B) edge [bend left] (A)
        (B) edge [bend left=20] (C)
        (B) edge [bend right] (D)
        (B) edge [loop below] (B)
        (C) edge [bend right] (A)
        (C) edge [bend left] (B)
        (D) edge (C)
        (D) edge [loop right] (B);
    \end{tikzpicture}
      \caption{Transition graph $\Upsilon$ for the CA rule
      $30$  where the region to 
      be controlled is of size $2$ and  $\lambda(0)=00$, $\lambda(1)=10$, $\lambda(2)=01$,
      $\lambda(3)=11$.}
        \label{figGC1}
    \end{center}
  \end{figure}
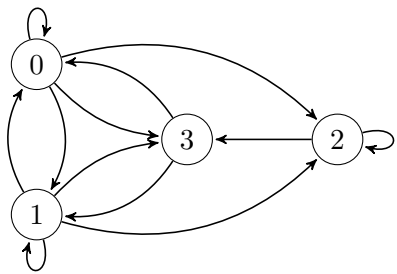

\begin{example}\textbf{Rule $90$ when the controlled region is of size $|\omega|=3$}.\\
Another example is given in Figure ~\ref{figGC2}. We have considered the rule $90$ where the controlled region is of size $3$. 

\begin{figure}[!ht]
  \begin{minipage}[h]{0.450\linewidth}
  
    \centering\includegraphics[scale=0.58]{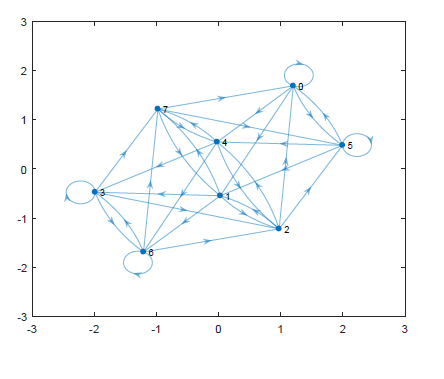}
    
  \end{minipage}
  \begin{minipage}[h]{0.50\linewidth}
    \[ \mathcal  C_{90}=
    \bordermatrix{
      ~ &  $0$ & $1$ & $2$ & $3$ & $4$ & $5$ & $6$ & $7$ \cr
      $0$ & $1$ &    $1$   &  $0$ &    $0$    & $1$  &   $1$ &    $0$ & $0$ \cr 
      $1$&  $0$ &    $0$ &    $1$ &    $1$    & $0$ &    $0$ &    $1$  &   $1$ \cr
      $2$ &  $1$  &  $1$   &$0$&     $0$&     $1$    & $1$ &    $0$     &$0$\cr
      $3$ & $0$  &   $0$  &   $1$     &$1$    & $0$  &   $0$ &    $1$&     $1$ \cr
      $4$ &   $0$ &    $0$    & $1$ &    $1$ &    $0$ &    $0$&     $1$ &    $1$ \cr
      $5$ & $1$   &  $1$ &    $0$ &    $0$ &  $1$ &    $1$  &   $0$   &  $0$\cr
      $6$ & $0$ &    $0$ &    $1$     &$1$ &   $0$    & $0$ &    $1$   &  $1$\cr
      $7$ & $1$ &    $1$ &    $0$ &    $0$    & $1$ &    $1$ &    $0$ &    $0$%
    },   
    \]
  \end{minipage}

  \caption{Transition graph and its adjacency matrix for the CA rule $90$ where the controlled region is of size $3$.}
  \label{figGC2}
  \end{figure}
\end{example}
\section{Characterizing regional controllability for Boolean deterministic CA}

\subsection{Necessary and sufficient  Condition: Hamiltonian Circuit}

In this section we prove the regional controllability for
one-dimensional and two-dimensional CA using a method  based on the
existence of a Hamiltonian circuit. The CA is regionally controllable if all the states are reachable in the target region (starting from each vertex we can achieve another vertex in finite number of steps ). The existence of a Hamiltonian circuit ensures that all vertices (configurations) are visited once and ensures that it exists a time $T$ such as all the configurations are reachable.

\begin{definition}\cite{rosen2017handbook}
  A Hamiltonian circuit of a graph $G=(V,A)$ is a simple directed path
  of $ G $ that includes every vertex exactly once.
\end{definition}

\begin{notation}
  We introduce the notation $a \leadsto b$. This  means that
   $b$ is reachable from $a$ \emph{i.e.} there  is a directed path starting from  $a$ to $b$. In other words, there
  exist vertices $v_1, v_2, \dots, v_i$ such that $(a,v_1), (v_1,
  v_2), \dots, (v_{i-1}, v_i), (v_i, b)$ are arcs in $A$.
\end{notation}

\begin{theorem}
  
  A Cellular Automaton is regionally controllable iff there exists a $t$ such that the graph associated to the transformation matrix $\mathcal{C}^{t}$ contains a Hamiltonian circuit.
\end{theorem}


\begin{proof}
Let us start with the first implication. 
  Let $\Upsilon=(V,A)$ be the transition graph built in Section \ref{sect:transG} for a CA with a controlled region of size $|\omega|$ and $V=\{v_{1},\dots,v_{2^{|\omega|}}\}$. The graph  $\Upsilon$ will be  represented by an adjacency matrix $\mathcal C$. Let $G_{1}$ be the transition graph associated to the matrix  $\mathcal{C}^{t} $. The proof is based on the following property in graph theory: the $(i,j)$th entry of the matrix $\mathcal{C}^t$ corresponds to the number of paths of length $t$ from vertex $i$ to $j$.
  
  Assume that the CA is regionally controllable at time $T \ge t$. Then some configurations can be reached in less than $T$ steps from any other one. That means that each pair of vertices are linked by a directed path of length at most equal to $T$. Therefore, $\mathcal{C}^T$ will be strictly positive as reported in the theorem in \cite{sdridi} which states that the CA is regionally controllable if there exists a power ${T}$ such that $\mathcal{C}^{T}>0$. The associated graph $G_T$ to the matrix $\mathcal{C}^T$ is therefore a complete graph and it is trivial to find a Hamiltonian circuit in a complete graph  which implies that there exists $t\le T$ such that the graph related to the matrix $\mathcal{C}^{t}$ contains a Hamiltonian Circuit and the direct implication  holds.\\

To prove the converse one, let us assume that  $G_{1}$ contains a Hamiltonian circuit. This means that there is a directed path that goes through all the vertices once. Therefore there exists an order $i_1, i_2, \dots, i_{2^{|\omega|}}$ such that: $(v_{i_1}, v_{i_2})$, $(v_{i_2}, v_{i_3})$, $\dots$, $(v_{i_{2^{|\omega|-1}}}, v_{i_{2^{|\omega|}}})$, $(v_{i_{2^{|\omega|}}}, v_{i_{1}})$ are arcs in $A$.
    And then we have: \[ v_{i}  \leadsto  v_{j}~~ \forall ~i,j\in\{1,\dots,2^{|\omega|}\} ~and ~i\neq j \] Thus, $\exists T \ge t$ such that all the vertices (configurations) are reachable. And the theorem holds.
\end{proof}

As the problem of proving the existence of a Hamiltonian circuit in a
graph is NP-complete, the time complexity can be exponential in the number of vertices of the transformation graph. We improve this criterion  in the next section with a solution in polynomial time that gives a necessary and sufficient condition.


\subsection{Necessary and Sufficient Condition: Strongly connected component}

\begin{definition}\cite{rosen2017handbook}
  A strongly connected component (SCC for short) of a directed graph $
  G $ is a maximal set of vertices $ C \subset V $ such that for every
  pair of vertices $ v_{1} $ and $ v_{2} $ in $C$, there is a directed
  path from $v_{1}$ to $v_{2}$ and a directed path from $ v_{2} $ to
  $v_{1}$.
\end{definition}

\begin{theorem}\label{thm:scc}
  A CA is regionally controllable for a given rule iff the   transition graph $\Upsilon$ associated to the rule has only one SCC.
\end{theorem}

\begin{proof}
  Let $\Upsilon=(V,A)$ be the transition graph built in Section
  \ref{sect:transG} from a controlled region of size $|\omega|$.

  Assume that the graph contains only one SCC. There exists a directed   path which relates each pair of vertices of the graph. Hence there is
  a sequence of controls that permits to go from every configuration
  to any other one. The CA is then controllable on $\omega$ and the direct implication holds.
  
  Assume now that the graph $\Upsilon$ contains more than one SCC, let
  say it contains two. Then, the set of vertices can be
  divided in two sets related to each SCC such as: \[
  V_{1}=\{v_{1},\dots,v_{k}\} \quad
  V_{2}=\{v_{k+1},\dots,v_{2^{|\omega|}}\}\] and there is no arc between
  $V_1$ and $V_2$.  Therefore, there is no control that allows  to
  obtain a configuration represented in $V_{1}$ from a configuration
  in $V_2$ according to the construction of $\Upsilon$. It is impossible since the CA is regionaly controllable and so the converse implication holds.
\end{proof}

\textbf{Time complexity} To find the SCCs, we have used Tarjan's algorithm \citep{Tarjan72depthfirst} which  has a linear
time complexity: $O(|V|+|A|)$ on the graph $\Upsilon=(V,A)$. If we consider a controlled region $\omega$ of size $|\omega|$ and since in that case $|A| \leq 4 |V|$, then the time complexity is $O(|V|) = O(2^{|\omega|})$. \\

\begin{remark}
The regional controllability depends on the rule and the size of the controlled region. The size of the controlled region for the same rule has an impact on the number of SCCs. According to Theorem \ref{thm:scc}, by changing the size of CA, a rule can be sometimes regionally controllable and sometimes not. \\


\end{remark}

In Table \ref{tab:1} the results of our simulations are highlighted. \\

\begin{table}[!ht]
    \centering
    \tbl{Classification of some rules of one-dimensional CA}
    {\begin{tabular}{|c|c|c|}\hline
	Rules & Decidability Criterion & number of SCC\\
	\hline\hline
	0,255 & not controllable & 64 for  $|\omega|=6$ and 16 for $|\omega| =4$ \\
	\hline
	1 & not controllable $|\omega|=4$, 
 controllable $|\omega|=2$ & 8 for $|\omega|=4$ and 1 for  $|\omega|=2$ \\
	\hline
	60,90,102,150,170 & controllable & 1\\
	\hline
	204 & not controllable & $ \bg 1$\\
	\hline 
	2,4,8,16 & not controllable & $ \bg 1$ \\
	\hline 
	105,195,165,153,85 & controllable for all the sizes of CA & 1\\
	\hline
	22 & controllable for $|\omega|=2$, not controllable $|\omega|=5$ & 1 for $|\omega|=2$ and $\bg 1$  for $|\omega|=5$\\
	\hline 
	26 & controllable $ 2\le |\omega|\le 3 $, not controllable $|\omega|=4$& 1 for $ 2\le |\omega|\le 3 $ , 2 for $|\omega|=4$\\
	\hline
	233 & not controllable & $ \bg 1$ \\
	\hline
	3 & controllable $|\omega|=2$, otherwise not controllable & 1 for $|\omega|=2 $ , otherwise  $ \bg 1$\\
	\hline
	\end{tabular}}

    \label{tab:1}
\end{table}

To illustrate the obtained results, we shall give in the following section, some examples in both  one and two dimensional  cellular automata.

\subsection{Examples}
\begin{example}

	Let us consider the two linear Wolfram rules  $ 150 $ and  $ 90 $ \cite{wolfram1984s}. For a controlled region of $ |\omega|= 4 $, the graph of the matrix $ \mathcal{C} $ associated to these two rules is illustrated in the  figure Fig.~\ref{SCC}. It contains one strongly connected component which means that there exists a time $ T  $ where each configuration is reachable.\\

\begin{figure}[!ht]
\label{SCC}
\centering
\subfloat[One SCC associated to the graph of rule 150.]{%
\resizebox*{7cm}{!}{\includegraphics{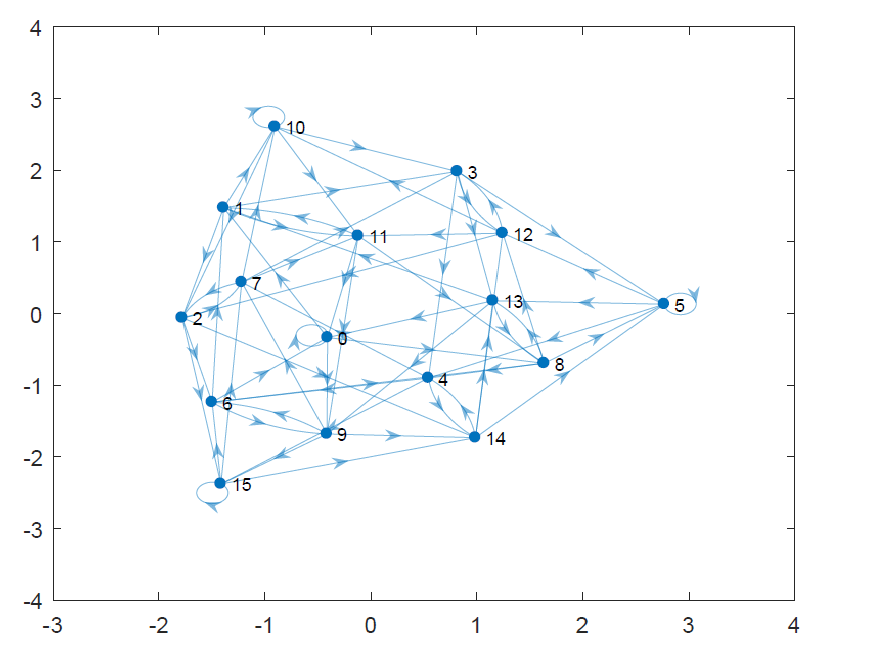}}}\hspace{5pt}
\subfloat[One SCC associated to the graph of rule 90.]{%
\resizebox*{7 cm}{!}{\includegraphics{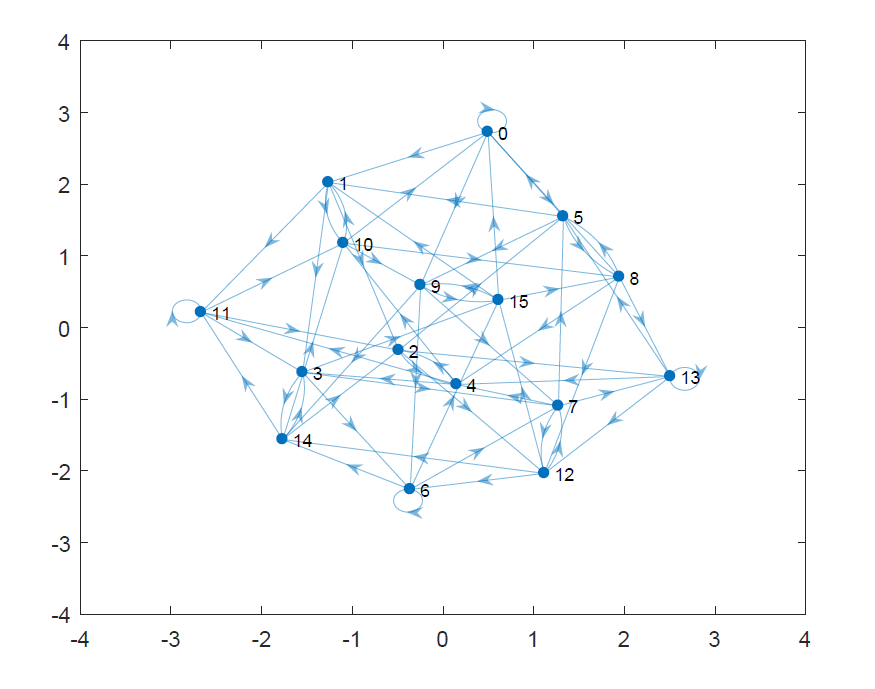}}}
\caption{Graphs of the matrices $\mathcal{C}_{150}$ and $\mathcal{C}_{90}$ respectively.} 
\end{figure}

	\end{example}
	
	\begin{example}
		Wolfram Rule $ 0 $ is not controllable neither its Boolean complement rule $ 255 $ as they  converge to a fixed point. The graph of their matrix $ \mathcal{C} $  contains more then one strongly connected component and the previous theorem states that these rules  are not regionally controllable for every region $\omega$.

		\begin{figure}[!ht]
\centering
\subfloat[rule 0, $|\omega|=6$.]{%
\resizebox*{7cm}{!}{\includegraphics{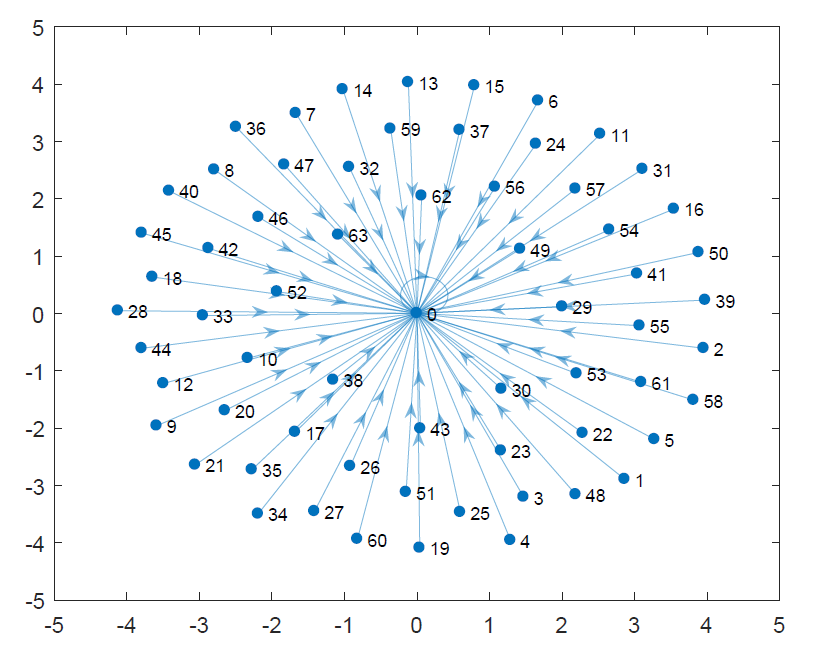}}}\hspace{5pt}
\subfloat[rule 255, $|\omega|=4$.]{%
\resizebox*{7 cm}{!}{\includegraphics{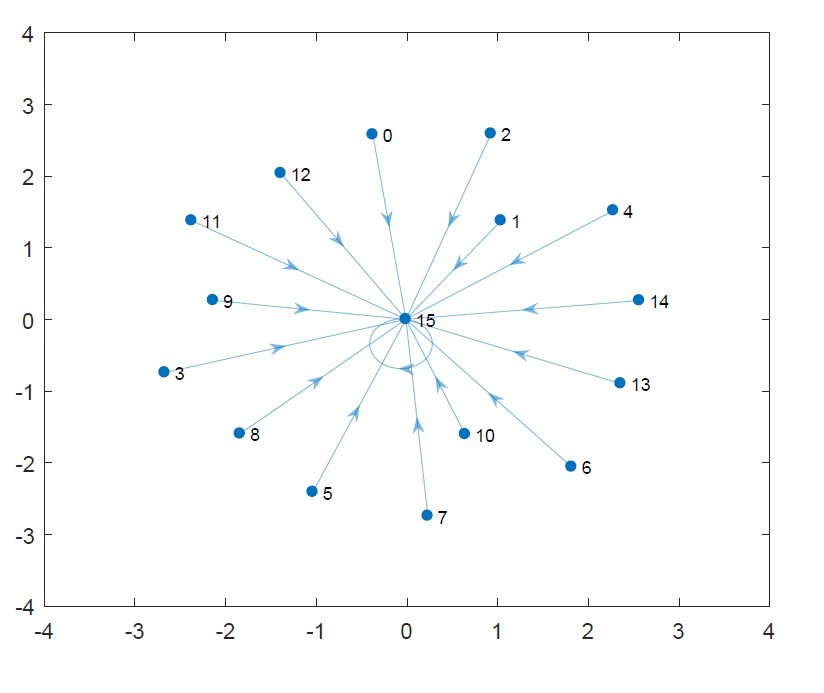}}}
\caption{Graphs related to the matrices $\mathcal{C}_0$ and $\mathcal{C}_{255}$ respectively.} 
\end{figure}

	\end{example}

	\begin{example}
		Let us consider now a two-dimensional cellular automaton. Its local evolution is given by the transition function: 
  \begin{displaymath}
    s^{t+1}(c_{i,j})=s^{t}(c_{i-1,j}) \oplus s^{t}(c_{i+1,j}) \oplus s^{t}(c_{i,j-1}) \oplus s^{t}(c_{i,j+1})
  \end{displaymath}
 that is also denoted by rule $170$ using Wolfram's formalism.  We impose asymmetric controls by setting all cells on the boundaries to $0$ except for the two  red colored cells illustrated in Fig.~\ref{asymetric}. 
  
  
 The obtained graph of the matrix $ \mathcal{C}_{170} $ for $ |\omega|=2 \times 2, \quad \omega=\{c_{1,1},c_{1,2},c_{2,1},c_{2,2}\} $,   contains one strongly connected component and so the CA is regionally controllable.
		
			\begin{figure}[!ht]
\centering
{%
\resizebox*{10cm}{!}{\includegraphics{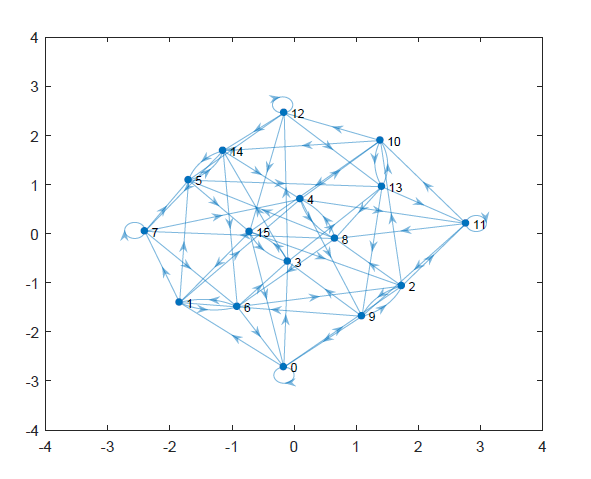}}}
\caption{Graph of the matrix $\mathcal{C}_{170}$ in two dimensional CA} 
\end{figure}
		
	\end{example}

\begin{example}
Finally, an  example with rule $1$ is given to show that the decidability criterion for regional controllability may change for the same rule, according to the size of  $\omega$. With a region of size $ |\omega|=6$, the graph of the matrix $ \mathcal{C}_1 $ contains more than one SCC while  for $ |\omega|=2 $ it contains only one strongly connected component. Consequently, the CA is not  regionally controllable in the first case and regionally controllable in the second one. 
	
		\begin{figure}[!ht]
\centering
\subfloat[rule 1, $|\omega|=6$.]{%
\resizebox*{7cm}{!}{\includegraphics{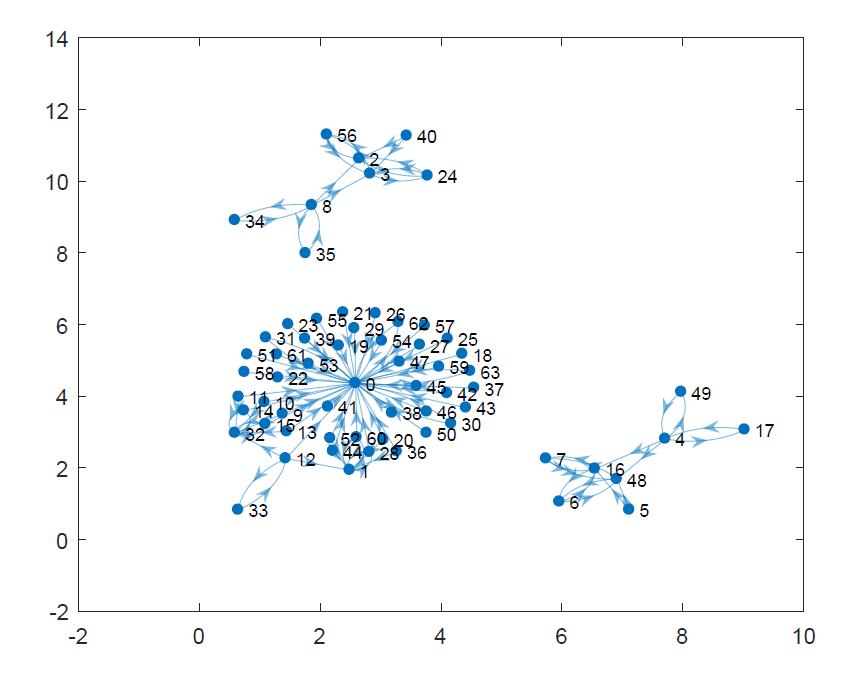}}}\hspace{5pt}
\subfloat[rule 1, $|\omega|=2$.]{%
\resizebox*{7 cm}{!}{\includegraphics{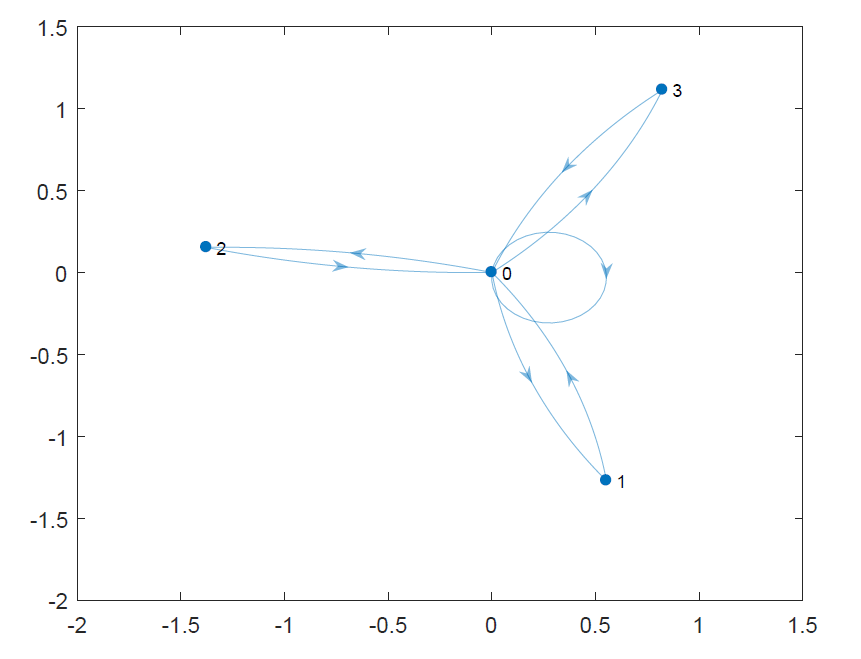}}}
\caption{Graphs of the matrix $\mathcal{C}_1$ for two sizes of $\omega$.} \label{sample-figure}
\end{figure}
\end{example}

\section{Pre-images of a regional controlled area}
Let $\{s_1^i,s_2^i, \dots, s_{n}^i\}$ be the configuration at time
$i$ of the region to be controlled. 
The idea is to find a boundary control given as a
sequence  $(\ell^0,r^0), (\ell^1,r^1), \dots, (\ell^{T-1},r^{T-1}$) so as to obtain a desired configuration $\{s_1^T,s_2^T, \dots, s_{n}^T\}$ at time $T$ 
from an initial one $\{s_1^0,s_2^0, \dots, s_{n}^0\}$. \\

Let us define  in what follows some needed notions and present the data structure required to solve the problem.  

\subsection{Distance function}

We define the distance function 
$$\Delta_i: \mbox{vertex} \mapsto \mbox{list~of~vertices}$$
that associates to each vertex $v \in \Upsilon$, the list of vertices from which $v$ can be reached within a path of length exactly $i$. 

Where $\Upsilon$  is the transition graph introduced in Section  \ref{sect:transG}. 

Therefore $\Delta_i(v)$ gives all initial configurations 
from which the desired configuration $v$ can be reached in exactly $i$ steps with the application of control.

Let $T$ be the time at 
which we want to reach a desired state. Representing  $\Delta$ as a map, we shall consecutively  construct the functions $\Delta_i$ by searching the predecessors  from $\Delta_{i-1}$. We start by $\Delta_1$ and go until  $\Delta_T$. The algorithm is given in the appendix.

If $\delta$ is the maximum number of predecessors of a vertex,  the time complexity is $O(T \times \delta \times |V|)=O(T \times \delta \times 2^{|\omega|})$.

\subsection{Path controllability}
We address two problems in this section.
\begin{problem}

Find one state configuration that can be driven to a desired state configuration  $b_{f}$ in $k$ steps and the relative control sequence.

\end{problem}

To solve this problem, we construct the distance function $\Delta_i$ for $1\leq i\leq k$. Then we consider the ancestors $b_i$ at distance $k$ of
$b_f$ (stored in $\Delta_k(b_f)$). If there is no ancestor, that means
that it is not possible to reach this state in $k$ steps. Otherwise,
 we can find the path of length $k$ with end extremity $b_f$. To do so, we
pick one predecessor of $b_f$, say $b_{k-1}$,
that can be reached in $k-1$ steps, \emph{i.e.} one among the vertices in the list
$\Delta_{k-1}(b_f)$. Then we search the first one among the
predecessors of $b_{k-1}$, say $b_{k-2}$, that can be reached in $k-2$
steps, \emph{i.e.}  one among the vertices in the list
$\Delta_{k-2}(b_f)$ and so forth until
we find $b_1$. We obtain the path of configurations $b_1, \dots,
b_{k-1}, b_f$. It just remains to find the appropriate control by
applying the rules to each configuration extended with the boundaries
$(0,0), (0,1), (1,0), (1,1)$.

In total, once the distance function is built, it takes $O(k)$ time. 

\begin{problem}

Find all the needed controls and the intermediate states of the controlled region required to obtain a desired state $b_f$ in exactly $k$ steps (\emph{i.e.}) from a state configuration $b_1$.


\end{problem}
For this problem, instead of checking if the list of ancestors is not
empty (and taking one among the vertices), we need to check if among the ancestors there is the initial configuration $b_1$. 

Therefore, the time complexity is $O(k \times \delta)$.

\section{Conclusion}
We have studied  in this paper the problem  of  regional controllability of  Boolean cellular automata focusing  on  actions performed on the boundary of a target region. We  established  some necessary and sufficient conditions using graph theory tools. We showed that  the existence of  a Hamiltonian circuit or   a strongly connected component guarantees  the regional controllability.  To obtain  the  control that allows the system to  reach the desired state during a given time horizon and starting from a given initial condition, an efficient algorithm for  generating  preimages was used.  Several examples of Elementary cellular automata has been considered. The obtained  control is not unique at this stage and the problem of optimality will be addressed afterward. A first problem of  regional controllability in minimal time is currently under study.

\section*{Acknowledgement(s)}
Sara DRIDI who was supported by the Algerian grant Averroes, would like to thank the Algerian government for this funding.

\bibliographystyle{tfnlm}
\bibliography{interactnlmsample}

\begin{thebibliography}{10}
\providecommand{\url}[1]{\normalfont{#1}}
\providecommand{\urlprefix}{Available from: }

\bibitem{cury2013systems}
Cury~JE, Baldissera~FL. Systems biology, synthetic biology and control theory:
  a promising golden braid. Annual Reviews in Control.
  2013;\hspace{0pt}37(1):57--67.

\bibitem{intriligator1980applications}
Intriligator~MD. The applications of control theory to economics. In: Analysis
  and optimization of systems. Springer; 1980. p. 605--626.

\bibitem{sontag2013mathematical}
Sontag~ED. Mathematical control theory: deterministic finite dimensional
  systems. Vol.~6. Springer Science \& Business Media; 2013.

\bibitem{liu2009elementary}
Liu~W. Elementary feedback stabilization of the linear
  reaction-convection-diffusion equation and the wave equation. Vol.~66.
  Springer Science \& Business Media; 2009.

\bibitem{street1995analysis}
Street~SE. Analysis and control of nonlinear infinite dimensional systems. IEEE
  Transactions on Automatic Control. 1995;\hspace{0pt}40(4):787.

\bibitem{zerrik2000}
Zerrik~E, Boutoulout~A, Jai~AE. Actuators and regional boundary controllability
  of parabolic systems. International Journal of Systems Science.
  2000;\hspace{0pt}31(1):73--82.

\bibitem{el1995regional}
El~Jai~A, Simon~M, Zerrik~E, et~al. Regional controllability of distributed
  parameter systems. International Journal of Control.
  1995;\hspace{0pt}62(6):1351--1365.

\bibitem{zerrik2004regional}
Zerrik~E, Ouzahra~M, Ztot~K. Regional stabilisation for infinite bilinear
  systems. IEE Proceedings-Control Theory and Applications.
  2004;\hspace{0pt}151(1):109--116.

\bibitem{lions1986exact}
Lions~J. Exact controllability of distributed systems. Comptes rendus de
  l'Acad\'emie des Sciences, S\'erie I-Math\'ematique.
  1986;\hspace{0pt}302(13):471--475.

\bibitem{toffoli1984cellular}
Toffoli~T. Cellular automata as an alternative to (rather than an approximation
  of) differential equations in modeling physics. Physica D: Nonlinear
  Phenomena. 1984;\hspace{0pt}10(1-2):117--127.

\bibitem{green1990cellular}
Green~D. Cellular automata models in biology. Mathematical and Computer
  Modelling: An International Journal. 1990;\hspace{0pt}13(6):69--74.

\bibitem{chopard2012cellular}
Chopard~B. Cellular automata modeling of physical systems. In: Computational
  complexity. Springer; 2012. p. 407--433.

\bibitem{bagnoli2017toward}
Bagnoli~F, El~Yacoubi~S, Rechtman~R. Toward a boundary regional control problem
  for {Boolean} cellular automata. Natural Computing. 2017;\hspace{0pt}:1--8.

\bibitem{el2003analyse}
El~Yacoubi~S, Jacewicz~P, Ammor~N. Analyse et contr{\^o}le par automates
  cellulaires. Annals of the University of Craiova-Mathematics and Computer
  Science Series. 2003;\hspace{0pt}30:210--221.

\bibitem{fekih2006regional}
Fekih~AB, El~Jai~A. Regional analysis of a class of cellular automata models.
  In: International Conference on Cellular Automata; Springer; 2006. p. 48--57.

\bibitem{el2002regional}
El~Yacoubi~S, El~Jai~A, Ammor~N. Regional controllability with cellular
  automata models. In: International Conference on Cellular Automata; Springer;
  2002. p. 357--367.

\bibitem{bagnoli2018regional}
Bagnoli~F, Dridi~S, El~Yacoubi~S, et~al. Regional control of probabilistic
  cellular automata. In: International Conference on Cellular Automata;
  Springer; 2018. p. 243--254.

\bibitem{sdridi}
Dridi~S, El~Yacoubi~S, Bagnoli~F. Boundary regional controllability of linear
  boolean cellular automata using {Markov} chain. Advances in Intelligent
  Systems and Computing. 2019;\hspace{0pt}.

\bibitem{wolfram1983}
Wolfram~S. Statistical mechanics of cellular automata. Reviews of modern
  physics. 1983;\hspace{0pt}55(3):601.

\bibitem{Wolfram2002}
Wolfram~S. A new kind of science. Vol.~5. Wolfram media Champaign, IL; 2002.

\bibitem{rosen2017handbook}
Rosen~KH. Handbook of discrete and combinatorial mathematics. Chapman and
  Hall/CRC; 2017.

\bibitem{Tarjan72depthfirst}
Tarjan~R. Depth first search and linear graph algorithms. Siam Journal of
  Computing. 1972;\hspace{0pt}1(2).

\bibitem{wolfram1984s}
Wolfram~S. Universality and complexity in cellular automata. Physica D:
  Nonlinear Phenomena. 1984;\hspace{0pt}10(1-2):1--35.

\end{thebibliography}

\newpage 

\appendix
We present in the appendix, the algorithms to construct the 
data structures used in the paper.

\ \\

Construction of the transition graph $\Upsilon$.

\fbox{\begin{minipage}{\textwidth}\setlength{\parskip}{5mm}
\begin{tabular}{ll}
\noindent\texttt{Algorithm transGraph($d_{\omega}$, $F$)}&\\
\hspace*{2em}\texttt{$d \leftarrow d_{\omega}$}&\\
\hspace*{2em}\texttt{$d_{\Upsilon} \leftarrow 2^{d}$}&\\
\hspace*{2em}\texttt{$\Upsilon$  $\leftarrow$ $[0]_{|d_{\Upsilon}\times d_{\Upsilon}}$} & \textit{(zero matrix of size $d_\Upsilon$)}\\
\hspace*{2em}\texttt{forall $0\leq i < d_{\Upsilon}$} & \textit{(for every configuration $i$)}\\
\hspace*{4em}\texttt{$\lambda(v_1) \leftarrow F(0\cdot \breve{i}|_{d} \cdot 0) $}&\\
\hspace*{4em}\texttt{$\lambda(v_2) \leftarrow F(0\cdot \breve{i}|_{d} \cdot 1) $}&\\
\hspace*{4em}\texttt{$\lambda(v_3) \leftarrow F(1\cdot \breve{i}|_{d} \cdot 0) $}&\textit{(add the boundary controls}\\
& \textit{and apply the rule)}\\
\hspace*{4em}\texttt{$\lambda(v_4) \leftarrow F(1\cdot \breve{i}|_{d} \cdot 1]) $}&\\
\hspace*{4em}\texttt{forall $0\leq j < d_{\Upsilon}$} & \textit{(for every configuration $j$)}\\
\hspace*{6em}\texttt{if $\lambda(v_1),\lambda(v_2),\lambda(v_3)$ or $\lambda(v_4)$ equal $\breve{j}_{|d}$} & \textit{(add an edge if there}\\
 & \textit{is a boundary control)}\\
\hspace*{8em}\texttt{$\Upsilon(i,j) \leftarrow 1$} &\textit{for i that leads to j)}\\
\hspace*{2em}\texttt{return $\Upsilon$}&\\
\end{tabular}
\end{minipage}}

\ \\

Construction of the distance function.

\fbox{\begin{minipage}{\textwidth}\setlength{\parskip}{5mm}
\noindent\texttt{Algorithm distanceFunction($G_\Upsilon$, k, v)}\\
\hspace*{2em}\texttt{Let $d_\Upsilon$ be the number of vertices of $G_\Upsilon$}\\
\hspace*{2em}\texttt{Let $\Delta$ be an empty matrix of size $d_\Upsilon \times k$}\\
\hspace*{2em}\texttt{for each vertex $v$ in $G_\Upsilon$}\\
\hspace*{4em}\texttt{$\Delta(v, 1)$ $\leftarrow$ predecessors($G_\Upsilon$, $v$)} \\
\hspace*{2em}\texttt{for $1<\ell\leq k$}\\
\hspace*{4em}\texttt{for each vertex $v$ in $G_\Upsilon$}\\
\hspace*{6em}\texttt{for each vertex $u$ in $\Delta(v, \ell - 1)$}\\
\hspace*{8em}\texttt{add($\Delta(v, \ell)$, predecessors($G_\Upsilon$, $u$))} \\
\hspace*{2em}\texttt{return $\Delta$}\\
\end{minipage}}

\ \\ 

Finding the control in $k$ steps to reach the desired state.

\fbox{\begin{minipage}{\textwidth}\setlength{\parskip}{5mm}
\noindent\texttt{Algorithm pathControllability($G_\Upsilon$, $k$, $v_{init}$, $v_{desired}$)}\\
\hspace*{2em}\texttt{$p \leftarrow [v_{desired}]$}\\
\hspace*{2em}\texttt{$\Delta \leftarrow $ distanceFunction$(G_\Upsilon, k, v_{desired})$}\\
\hspace*{2em}\texttt{if $v_{init} \notin \Delta(v, k)$}\\
\hspace*{4em}\texttt{return $p$} \\
\hspace*{2em}\texttt{pred $\leftarrow v_{desired}$}\\
\hspace*{2em}\texttt{for $i$ from $k$ to $1$}\\
\hspace*{4em}\texttt{predList $\leftarrow$ predecessors($G_\Upsilon$, $pred$)}\\
\hspace*{4em}\texttt{pred $\leftarrow$ find($x \in$ predList$. v_{init} \in \Delta(x, i)$)}\\
\hspace*{4em}\texttt{add(pred, $p$)} \\
\hspace*{2em}\texttt{return $p$}\\
\end{minipage}}

\end{document}